\documentclass[a4paper,12pt]{article}
\usepackage{a4wide}
\usepackage{amssymb, amsmath}
\usepackage[amsmath,hyperref,thmmarks]{ntheorem}
\usepackage{t1enc,times,graphicx}
\usepackage[utf8]{inputenc}
\usepackage[english]{babel}
\usepackage{enumerate}
\usepackage{hyperref}
\usepackage{color,caption,subcaption}
\usepackage{natbib}
\graphicspath{{results2/pnl}}

\emergencystretch=\hsize
\def\ca{{\mathcal A}}

\def\cf{{\mathcal F}}

\def\ch{{\mathcal H}}

\def\cp{{\mathcal P}}

\def\ct{{\mathcal T}}

\def\E{{\mathbb E}}
\def\H{{\mathbb H}}

\def\L{{\mathbb L}}
\def\N{{\mathbb N}}
\def\P{{\mathbb P}}

\def\R{{\mathbb R}}

\def\s{\star}

\def\ind#1{{\bf 1}_{\left\{#1\right\}}}

\def\Var{\mathop{\rm Var}\nolimits}

\def\Span{\mathop{\rm Span}\nolimits}

\def\inv#1{\mathop{\frac{1}{ #1}}\nolimits}

\theoremstyle{plain}
\newtheorem{theorem}{Theorem}[section]
\newtheorem{proposition}[theorem]{Proposition}

\newtheorem{remarkh}[theorem]{Remark}

\theoremstyle{nonumberplain}
\theoremheaderfont{\normalfont\itshape}
\theorembodyfont{\normalfont}
\theoremseparator{.}
\theoremsymbol{\ensuremath{\blacksquare}}
\newtheorem{proof}{Proof}

\makeatletter
\newcounter{hypo}
\renewcommand{\thehypo}{(${\mathcal H}$-\arabic{hypo})}
\newcommand{\dohypo}{\thehypo}

\makeatother

\title{How can the dual martingale help solving the primal optimal stopping problem? }
\date{\today}
\author{Aurélien Alfonsi\thanks{ CERMICS, ENPC, Institut Polytechnique de Paris, CNRS, Marne-la-Vallée, France \& MathRisk team-project, Inria Paris, France. \newline \texttt{aurelien.alfonsi@enpc.fr}} \and Ahmed Kebaier\thanks{LaMME, CNRS, UMR 8071, Université \'Evry Paris Saclay, 91037, \'Evry, France. \newline \texttt{ahmed.kebaier@univ-evry.fr}} 
\and Jérôme Lelong\thanks{ Univ. Grenoble Alpes, CNRS, Grenoble INP, LJK, 38000 Grenoble, France. \newline\texttt{jerome.lelong@univ-grenoble-alpes.fr} \newline{Acknowledgements: AA and AK benefited from the support of the “chaire Risques financiers”, Fondation du Risque.}}}

\begin{document}
\maketitle
\begin{abstract}
Motivated by recent results on the dual formulation of optimal stopping problems, we investigate in this short paper how the knowledge of an approximating dual martingale can improve the efficiency of primal methods. In particular, we show on numerical examples that accurate approximations of a dual martingale efficiently reduce the variance for the primal optimal stopping problem.
\end{abstract}
\noindent {\bf Keywords:} Optimal stopping, Variance Reduction, Pure dual algorithm, Martingale, Least square Monte Carlo, Bermudan option.\\
\noindent {\bf AMS subject classifications (2020):} 62L15, 60G40, 91G20, 65C05.

\section{Introduction and framework }

Let $N\in \N^*$ and $(\Omega, \overline{\cf}, \cf=(\cf_n)_{0 \le n \le N}, \P)$ be a filtered probability space with a discrete time filtration.
We consider a market with $d$ assets $(S^k_n,n\ge 0)$, $k\in \{0,\dots, d\}$, with $S^0$ being the risk-free asset.  We are interested in the pricing and hedging of a Bermudan option paying $\Psi(S^1_{n},\dots,S^d_{n})$ if exercised at time $n$. We denote by
\begin{equation}\label{eq:Z_Bermudan}
  Z_n=\frac{\Psi(S^1_{n},\dots,S^d_{n})}{S^0_n}
\end{equation}
the discounted payoff and we assume that  $\max_{0\le n\le N}\E[|Z_n|^2]<\infty$. 
We define the Snell envelope, for $n\in \{0,\dots,N\}$,
\begin{equation}
  \label{eq:price}
  U_{n} = \sup_{\tau \in \ct_{n, N}} \E[ Z_{\tau} | \cf_{n}],
\end{equation}
where $\ct_{n,N}$ denotes the set of $\cf$-stopping times taking values in $\{n,\dots,N\}$.
It classically solves the  dynamic programming equation
\begin{equation}
  \label{eq:dpp}
  \begin{cases}
    U_{N} & = Z_{N} \\
    U_{n} & = \max\left( Z_{n}, \E[U_{n+1} | \cf_{n}] \right), \quad 0 \le n  \le N-1.
  \end{cases}
\end{equation}
The process $U$ is a $L^2$ supermartingale and $\tau^\s:=\inf \{n \in \{0,\dots, N\}: U_n=Z_n \}$ is an optimal stopping time for~\eqref{eq:price} satisfying that $(U_{n\wedge \tau^\s})_{0 \le n \le N}$ is a martingale.
The optimal policy~$\tau^\s$   is obtained from the celebrated Longstaff Schwarz algorithm (see~\cite{LS})
\begin{equation}
  \label{eq:dpp-policy}
  \begin{cases}
    \tau_{N} & = N \\
    \tau_{n} & = n \ind{Z_n \ge \E[Z_{\tau_{n+1}} | \cf_{n}]} + \tau_{n+1} \ind{Z_n < \E[Z_{\tau_{n+1}} | \cf_{n}]}, \quad 0 \le n  \le N-1.
  \end{cases}
\end{equation}
Then, $\tau_n$ is the optimal stopping time on the time interval $\{n,\dots,N\}$, and satisfies $\sup_{\tau \in \ct_{n, N}} \E[ Z_{\tau} | \cf_n] = \E[Z_{\tau_n} | \cf_n]$. The stopping time $\tau_0$ is thus equal to $\tau^\s$.

Equations~\eqref{eq:price}, \eqref{eq:dpp} and~\eqref{eq:dpp-policy} are often referred to as the primal approach.  The dual approach was introduced by~\cite{rogers-02,HK} and consists in writing 
\begin{align}
  \label{eq:dual-price-k}
  U_{n} & = \inf_{M \in \H^2} \E\left[ \max_{n \le j \le N} \{Z_{j} - (M_{j}- M_{n})\}\bigg| \cf_{n}\right],
\end{align}
where $\H^2$ is the set of square integrable $\cf$-martingales.
The infimum is achieved, among others, by the martingale~$M^\s$ coming from the  Doob-Meyer decomposition
\begin{equation}
  \label{eq:doob-meyer}
  U_{n} = U_0 + M^\s_{n} - A^\s_{n},
\end{equation}
where $M^\s\in \H^2$ vanishes at $0$ and $A^\s$ is a predictable, nondecreasing and square integrable process also vanishing at $0$. There are many works in the literature proposing approximations of the martingale $M^\s$. For instance, \cite{HK}, \cite{AB}  and~\cite{BBS} construct these approximations by first numerically solving the primal problem:  \cite{HK} and~\cite{BBS} need an approximation of~$U$, while~\cite{AB} need an approximation of $\tau^\s$. In contrast, many other papers such as \cite{DFM}, \cite{BHS} and~\cite{Lelong} propose to directly solve the dual problem~\eqref{eq:dual-continuation-value} to get an approximation of one solution of~\eqref{eq:dual-price-k}. Recently, we have proposed in \cite{AKL} a new purely dual formulation, the only solution of which is $M^\s$. It is approximated by using only tradable instruments, which makes it also relevant for hedging.  

The aim of this short paper is to investigate how the knowledge of this approximation may help to solve the primal problem~\eqref{eq:dpp}. In particular, we propose a  new variance reduction technique motivated by the following remark: $Z_{\tau^\s}-M^\s_{\tau^\s}=U_{\tau^\s}-M^\s_{\tau^\s}=U_0$ is deterministic. In fact, we have $U_0=\E[Z_{\tau^\s}]= \E[U_{\tau^\s}]=U_0-\E[A^\s_{\tau^\s}]$, which gives that $A^\s_{\tau^\s}=0$ almost surely. Thus, if both $M^\s$ and $\tau^\s$ were exactly known, only one simulation would be necessary to compute $U_0$, which is the paramount of variance reduction! Now, suppose that we have at hand a martingale approximation $\hat{M}$ of $M^\s$ and that $\hat{\tau}$ is an approximation of $\tau^\s$, then a natural control variate reduction method is to compute a Monte-Carlo approximation of $\E[Z_{\hat{\tau}}-\lambda \hat{M}_{\hat{\tau}}]$, where $\lambda$ minimizes the variance of $Z_{\hat{\tau}}-\lambda \hat{M}_{\hat{\tau}}$. The idea of finding such a martingale for variance reduction purposes has already been considered in the literature, see e.g.~\cite{JuKa}. Here, we want to leverage on our recent results in~\cite{AKL} that presents a generic algorithm to approximate $M^\s$.

The paper is structured as follows. In Section~\ref{sec:main}, we present different ways of using the control variate method based on the approximated dual martingale. Then, we show on different numerical examples the relevance of the method, which provides a much more accurate estimator of~$U_0$. We also investigate if directly solving the Longstaff Schwartz algorithm on $Z-\hat{M}$ rather than on $Z$ may improve the accuracy. Up to our knowledge, this direction has not been much investigated in the literature, but it turns out in our case that it has little effect on the pricing accuracy.  In Section~\ref{sec:proxy}, we propose a random time approximation of the optimal policy that uses only $\hat{M}$. Last, the appendix sections recall the approximation $\hat{M}$ developed in~\cite{AKL} and also analyse the variance decomposition in the Longstaff Schwartz algorithm.


\section{The dual martingale  as a control variate for the primal problem}\label{sec:main}
\subsection{Methodology}

We start by a simple remark regarding the primal problem. Let us consider $M \in \H^2$, such that $M_0 = 0$. It is clear from the optional stopping theorem, that $\sup_{\tau \in \ct_{0, N}} \E[ Z_{\tau}] = \sup_{\tau \in \ct_{0, N}} \E[ Z_{\tau} - M_{\tau}]$. Hence, the Bermudan option price~$U_0$  can also be obtained by applying the Longstaff Schwartz algorithm either to the discounted payoff $Z_n$ or alternatively to $Z_n - M_n$.

A first idea would be then to use the martingale $\hat{M}$ developed in~\cite{AKL} (see also~\eqref{eq:hat_M}) and to implement the Longstaff Schwartz algorithm on the option with discounted payoff $Z - \hat{M}$. Note that as $\hat{M}$ is an approximation of $M^\s$, we expect that the least squares approximation of $\E[Z_{\tau_{n+1}} - \hat{M}_{\tau_{n+1}}| \cf_{n}]$ will have very little variance. Let us be more precise and denote by $\cp_n$ the $L^2$ projection on a finite dimensional vector subspace of $L^2(\Omega,\cf_n)$, which is used for the practical implementation of  the Longstaff-Schwartz algorithm (more generally, $\cp_n$ may be any function from $L^2(\Omega,\cf_N)$ to $L^2(\Omega,\cf_n)$ approximating the conditional expectation). Then, \eqref{eq:dpp-policy} is approximated by
\begin{equation*}
  \begin{cases}
    \tau^{LS'_2}_{N} & = N \\
    \tau^{LS'_2}_{n} & = n \ind{Z_n-\hat{M}_n \ge \cp_n \left(Z_{\tau^{LS'_2}_{n+1}}-\hat{M}_{\tau^{LS'_2}_{n+1}}\right)} + \tau^{LS'_2}_{n+1} \ind{Z_n-\hat{M}_n < \cp_n\left(Z_{\tau^{LS'_2}_{n+1}}-\hat{M}_{\tau^{LS'_2}_{n+1}} \right)},
  \end{cases}
\end{equation*}
for  $0 \le n  \le N-1$, but also by
\begin{equation}\label{eq:def_tauLS2}
  \begin{cases}
    \tau^{LS_2}_{N} & = N \\
    \tau^{LS_2}_{n} & = n \ind{Z_n \ge \cp_n\left(Z_{\tau^{LS_2}_{n+1}}-\hat{M}_{\tau^{LS_2}_{n+1}} +\hat{M}_n\right)} + \tau^{LS_2}_{n+1} \ind{Z_n < \cp_n\left(Z_{\tau^{LS_2}_{n+1}}-\hat{M}_{\tau^{LS_2}_{n+1}}+\hat{M}_n \right)},
  \end{cases}
\end{equation}
which are not equivalent since we typically have $\cp_n(\hat{M}_n) \not = \hat{M}_n$. In practice, for Markovian models, i.e. when $S$ introduced in~\eqref{eq:Z_Bermudan} is a Markov chain, we usually take $\cp_{n}$ as the projection on the vector space generated by polynomial functions of~$S_n$ with a bounded degree. To underline this, note that we do not have $Z_n-\hat{M}_n\ge \E[Z_{\tau_{n+1}}-\hat{M}_{\tau_{n+1}}|S_n]\iff Z_n \ge \E[Z_{\tau_{n+1}}|S_n]$ since in general $\E[\hat{M}_{n}|S_n] \not = \hat{M}_n$, but we do have $Z_n\ge \E[Z_{\tau_{n+1}}-\hat{M}_{\tau_{n+1}}+\hat{M}_n|S_n]\iff Z_n \ge \E[Z_{\tau_{n+1}}|S_n]$ since $\E[\hat{M}_{\tau_{n+1}}-\hat{M}_n|S_n] = 0$. This formally explains why it is much better to use~\eqref{eq:def_tauLS2} to approximate~\eqref{eq:dpp-policy}\footnote{For example on the Bermudan put case of Table~\ref{Table:Put_local}, with the same parameters as in the first line of the table, we obtain with $\tau_0^{LS'_2}$ a Bermudan put price of 9.5831 (below than the European price 9.6642 !) instead of 9.9064 with $\tau_0^{LS_2}$, which is closer to the real price since the price is approximated from below. Thus, in all our numerical experiments we will consider~\eqref{eq:def_tauLS2}.}. We will simply denote by $\tau^{LS_2}=\tau^{LS_2}_{0}$ the approximation of the optimal stopping time given by the algorithm. We will compare this stopping time with the classical stopping time computed on~$Z$ and defined by
\begin{equation}\label{eq:def_tauLS1}
  \begin{cases}
    \tau^{LS_1}_{N} & = N \\
    \tau^{LS_1}_{n} & = n \ind{Z_n \ge \cp_n\left(Z_{\tau^{LS_1}_{n+1}}\right)} + \tau^{LS_1}_{n+1} \ind{Z_n < \cp_n\left(Z_{\tau^{LS_1}_{n+1}} \right)}, \quad 0 \le n  \le N-1,
  \end{cases}
\end{equation}
and set $\tau^{LS_1}=\tau^{LS_1}_{0}$.

In practice, the projection $\cp_n$ can barely be computed in a closed form. We denote by $\hat{\cp}_n^Q$ its empirical approximation using $Q$ samples of $Z$ (and $\hat{M}$ if needed). This leads to the following implementable approximation of $(\tau^{LS_1}_n)_{0 \le n \le N}$
\begin{equation*}
  \begin{cases}
    \hat{\tau}^{LS_1,Q}_{N} & = N \\
    \hat{\tau}^{LS_1,Q}_{n} & = n \ind{Z_n \ge \hat{\cp}_n^Q\left(Z_{\hat{\tau}^{LS_1,Q}_{n+1}}\right)} + \hat{\tau}^{LS_1,Q}_{n+1} \ind{Z_n < \hat{\cp}_n^Q\left(Z_{\hat{\tau}^{LS_1,Q}_{n+1}} \right)}, \quad 0 \le n  \le N-1.  \end{cases}
\end{equation*}
We define $\hat{\tau}^{LS_2,Q}_{N}$ in a similar way.

We advise to use three different sets of samples to avoid potential issues of bias and overfitting that are already well known in the Longstaff Schwartz pricing method alone. The first sample is used to compute the coefficients $\alpha^{Q_1}_n$ for $n \in \{1, \dots, N\}$, which define the approximation~$\hat{M}_n$, see Eq.~\eqref{def_alpha_n_Q} and~\eqref{eq:hat_M}. Then, we implement the Longstaff Schwartz algorithm on the second one to get $\hat{\tau}^{LS_i, Q_2}$, $i\in \{1,2\}$. Then, we generate a third sample to calculate out-of-sample Monte-Carlo prices to approximate $\E[Z_{\hat{\tau}^{LS_i,Q_2}}]$. More precisely, we use here a control variate variance reduction and compute 
\begin{equation}\label{eq:prix_empirique}
\frac 1 {Q_3} \sum_{q=1}^{Q_3} Z^{(q)}_{\hat{\tau}^{LS_i,Q_2,(q)}} -\hat{\lambda} \hat{M}^{(q)}_{\hat{\tau}^{LS_i,Q_2,(q)}},
\end{equation} where $\hat{\lambda}=\frac{\sum_{q=1}^{Q_3} Z^{(q)}_{\hat{\tau}^{LS_i,Q_2,(q)}} \hat{M}^{(q)}_{\hat{\tau}^{LS_i,Q_2,(q)}}}{ \sum_{q=1}^{Q_3} (\hat{M}^{(q)}_{\hat{\tau}^{LS_i,Q_2,(q)}})^2 }$  minimizes the asymptotical variance. Here, we denote by $Q_3$ the number of samples in the second and third sets that are used for the Longstaff Schwartz algorithm and the Monte Carlo estimation. 

However, let us note that for a given martingale $\hat{M}$, the statistical error comes from two different sources: the first one is linked to the approximation of $\tau^{LS_i}$ by $\hat{\tau}^{LS_i,Q_2}$ and the second one is the classical statistical error associated with the computation of $\E[Z_{\tau^{LS_i}}]$. This is analysed in~\cite[Theorem 4.2]{CLP} where a CLT is shown, see Appendix~\ref{app:variance} for further discussion.

\subsection{Numerical examples}

\subsubsection{One dimensional options}
\label{sec:var-1d}

We start our discussion by considering a Bermudan put option with payoff function
$$ \Psi(S) = (K - S)_+.$$
The calculation of the approximating martingale~$\hat{M}$ is carried out by using the Algorithm presented in~\cite{AKL}. Appendix~\ref{Sec:CalcM} recalls the main steps of this algorithm and its parameters: $\bar{N}$ is the number of subticks, $P$ the number of local functions used in the regression and $Q_1$ is the number of samples. In Table~\ref{Table:Put_local}, we have reported the prices obtained by three different primal methods.
\begin{enumerate}
  \item The classical Longstaff Schwartz method without any variance reduction. The price and the standard deviation of the Monte-Carlo estimator are given in the caption (multiplied by $1.96$, this standard deviation gives the half-width of the 95\% asymptotic confidence interval).
  The standard deviation is estimated from 40 independent runs of the Longstaff Schwartz algorithm. 
  \item  The Longstaff Schwartz method applied to~\eqref{eq:def_tauLS2} to determine the optimal stopping rule, and then the control variate reduction method $Z_{\hat{\tau}^{LS_2,Q_2}}-\lambda \hat{M}_{\hat{\tau}^{LS_2,Q_2}}$ to calculate the price. The price and the standard deviation (computed on 40 independent runs) of the Monte-Carlo estimator are given in the 4th and 5th columns.
  \item The classical Longstaff Schwartz method applied to~\eqref{eq:def_tauLS1} to determine the optimal stopping rule, and then the control variate reduction method $Z_{\hat{\tau}^{LS_1,Q_2}}-\lambda \hat{M}_{\hat{\tau}^{LS_1,Q_2}}$ to calculate the price. The price and the standard deviation (computed on 40 independent runs) of the Monte-Carlo estimator are given in the 6th and 7th columns.
\end{enumerate}  
The last two columns indicate the price and standard deviation given by the dual method presented in~\cite{AKL}. This indicates the quality of the approximation of $M^\s$: heuristically, the smaller is the dual price the better is the martingale $\hat{M}$.


\begin{table}[h!]
  \centering
  \begin{tabular}{|rrr|rr|rr||rr|}
\hline 
      $Q_1$ &  $\bar{N}$\phantom{$\Big|$} &  $P$ &      LS Eq.~\eqref{eq:def_tauLS2} & stddev &   LS Eq.~\eqref{eq:def_tauLS1}  & stddev &     Dual & stddev \\
\hline
 $10^5$ &  1 & 50 &   9.9064 & 0.0092 &        9.9025 & 0.0100 &  10.3159 & 0.0087 \\
 $10^5$ &  5 & 50 &   9.9065 & 0.0045 &        9.9045 & 0.0060 &  10.0787 & 0.0050 \\
$2\cdot10^6$ & 20 & 50 &    9.9072 & 0.0007 &        9.9071 & 0.0005 &   9.9625 & 0.0006 \\
\hline
\end{tabular}
  \caption{Prices for a put option using a basis of $P$ local functions with $K = S_0 = 100$, $T = 0.5$, $r=0.06$, $\sigma=0.4$ and $N=10$ exercising dates. Longstaff Schwartz algorithm is run with $Q_2=Q_3=50000$ samples and  with a polynomial approximation of order 6. The LS price without variance reduction is {\bf $9.90$}, with a {\bf standard deviation} $\mathbf{0.0513}$. The value of $\lambda$ in the control variate method is around $0.99$ for each case.    \label{Table:Put_local}}
\end{table}
First, we observe that the use of the martingale~$\hat{M}$ indeed allows to significantly reduce the variance, from  a factor between $25$ up to $10^4$ depending on the accuracy of $\hat{M}$. Second, we notice that there is almost no difference between using the Longstaff Schwartz algorithm with~\eqref{eq:def_tauLS2}  or~\eqref{eq:def_tauLS1}: the main variance gain is given by using the control variate method afterwards. In the next examples, we have observed exactly the same thing, so we only present the results with $\hat{\tau}^{LS_1}$. However, we have noticed that the regressors of the Longstaff Schwartz method are often less noisy when using~\eqref{eq:def_tauLS2} instead of \eqref{eq:def_tauLS1}, but this has almost no incidence on the price computation and its variance, see Appendix~\ref{app:variance}. Last, we observe on this example that the optimal value of $\lambda$ is around $0.99$. This is not surprising, since the optimal choice would be $1$ if $\hat{M}$ were exactly equal to $M^\s$. The value of $\lambda$ is somehow another indicator to measure how $\hat{M}$ is close to $M^\s$.

Now, we consider a butterfly option with payoff function \begin{equation}\label{psi_butterfly}\Psi(S)=2\left(\frac{K_1+K_2}2-S\right)_+ - \left(K_1-S\right)_+ - \left(K_2-S\right)_+.\end{equation}
We have reported our results in Table~\ref{tab:butterfly_loc}. We observe a variance reduction\footnote{In all the tables, we report the standard deviation of the empirical Monte Carlo estimator that gives the precision on the price. The (multiplicative) computational  gain is given by the variance ratio.} up to $100$ on this example. The parameter $\lambda$ used in the control variate method is between $0.97$ and $0.99$ in this experiment. 
\begin{table}[htbp!]
  \centering
  \begin{tabular}{|rrr|rrc||rr|}
    \hline 
    $Q_1$ &  $\bar{N}$\phantom{$\Big|$} &   $P$ & LS Eq.~\eqref{eq:def_tauLS1} & stddev & $\lambda$ &   Dual & stddev \\
    \hline
    $10^5$ &  5 &  50 &      5.6568 & 0.0045 & 0.9743 & 6.1278 & 0.0034 \\
    $5\cdot 10^5$ & 20 &  50 &      5.6563 & 0.0017 & 0.9905 & 5.8726 & 0.0011 \\
    \hline
  \end{tabular}

  \caption{Prices for a butterfly option with parameters $K_1=90$ $K_2=110$, $S_0=95$, $T=0.5$, $r=0.06$, $\sigma=0.4$ and $N=10$, using a basis of~$P$ local functions. The Longstaff-Schwartz algorithm is run with order $6$ polynomials and $Q_2=Q_3=50000$ samples. The LS price without variance reduction is $5.65$, with a {\bf standard deviation} $\mathbf{0.015}$. \label{tab:butterfly_loc}}
\end{table}

\subsubsection{Bermudan options on many assets}

\paragraph{Basket option}

We consider a put option on a basket of assets with payoff function
\[ \Psi(S)=\left(K - \inv{d}\sum_{i = 1}^d S^i \right)_+, \  S \in \R^d. \]
We report in Table~\ref{tab:put3d} the price of a Bermudan option with this payoff for $d=3$. We notice that the variance reduction ranges from a factor $25$ for the martingale $\hat{M}$ obtained with $Q_1=10^5$ to a factor $400$ for the one obtained with $Q_1=5\cdot10^5$.

\begin{table}[htbp!]
  \centering
\begin{tabular}{|rrr|rrc||rr|}
\hline
 $Q_1$ &  $\bar{N}$\phantom{$\Big|$} &   $P$ & LS Eq.~\eqref{eq:def_tauLS1} & stddev & $\lambda$ &   Dual & stddev \\
\hline
$10^5$ &  1 & 50 &      4.0373 & 0.0049 & 0.9782 & 4.3479 & 0.0045 \\
$2.5\cdot10^5$ &  5 & 50 &      4.0415 & 0.0015 & 0.9839 & 4.1517 & 0.0016 \\
$5\cdot10^5$ & 10 & 50 &      4.0424 & 0.0009 & 0.9855 & 4.1148 & 0.0009 \\
\hline
\end{tabular}
    \caption{Prices for a Basket option with $d=3$, $\sigma=0.2$, $r=0.05$, $T=1$, $N=10$, $S^1_0=S^2_0=S^3_0=100$, $K=100$.  LS implemented with polynomials function of order $5$, $Q_2=Q_3=50000$. The LS price without variance reduction is $4.03$, with a {\bf standard deviation} $\mathbf{0.021}$.  \label{tab:put3d}}
\end{table}

\paragraph{Max-call option}

We consider a call option on the maximum of a basket of assets with payoff:
\[ \Psi(S)= \left(K - \max_{1 \le k \le d} S^k\right)_+.\]
We have reported in Table~\ref{tab:maxcall} the price of a Bermudan option with 2 assets. This example has been considered in~\cite[Table 2]{AB}. We observe a variance reduction factor from $100$ to $1600$, which is very satisfactory. We note that the value obtained with the dual algorithm is rather far from the primal value, and  the optimal value of $\lambda$ is around $0.95$. These facts indicate that $\hat{M}$ is quite different from $M^\s$, nonetheless the variance reduction method still works well. 

\begin{table}[htbp!]
  \centering
  \begin{tabular}{|rrr|rrc||rr|}
\hline
 $Q_1$ &  $\bar{N}$\phantom{$\Big|$} &   $P$ & LS Eq.~\eqref{eq:def_tauLS1} & stddev & $\lambda$ &   Dual & stddev  \\
\hline
$10^6$ &  1 & 10 &       8.0675 & 0.0038 & 0.9357 & 8.9877 & 0.0041 \\
$2\cdot 10^6$ &  5 & 10 &       8.0671 & 0.0018 & 0.9480 & 8.5439 & 0.0018 \\
$4\cdot 10^6$ & 10 & 10 &       8.0675 & 0.0012 & 0.9501 & 8.4753 & 0.0011 \\
\hline
\end{tabular}

    \caption{Prices for a max-call option with $d=2$, $K=100$, $S^1_0=S^2_0=90$, $\sigma=0.2$, $\delta=0.1$, $r=0.05$, $T=3$, $N=9$. LS parameters:  polynomial functions of order $5$, $Q_2=Q_3=50000$. The LS price without variance reduction is $8.06$, with a {\bf standard deviation} $\mathbf{0.044}$.\label{tab:maxcall}}
\end{table}


\paragraph{Min Butterflies}

For the last example, we consider a rather exotic option that pays
$$\min_{1\le i\le d}(\Psi(S^i)),$$
where $\Psi$ is the butterfly payoff defined in~\eqref{psi_butterfly}.
Table~\ref{tab:minbutter} reports the prices for such an option on $d=2$ assets. Again, the variance reduction factor is important and ranges from 16 to 400, according to the effort made to calculate~$\hat{M}$. We note that $\lambda \approx 0.93$ is quite far from~$1$ when $M^\s$ is roughly approximated, but gets very close to~$1$ when more computational effort is made for the calculation  of $\hat{M}$.  

\begin{table}[htbp!]
  \centering
 \begin{tabular}{|rrr|rrc||rr|}
\hline
 $Q_1$ &  $\bar{N}$\phantom{$\Big|$} &   $P$ & LS Eq.~\eqref{eq:def_tauLS1} & stddev & $\lambda$ &   Dual & stddev   \\
\hline
$3\cdot 10^5$ &  1 & 10 &      2.1940 & 0.0031 & 0.9292 & 2.6744 & 0.0032  \\
$10^6$ &  5 & 10 &      2.1943 & 0.0013 & 0.9784 & 2.4506 & 0.0012  \\
$4\cdot 10^6$ & 10 & 10 &      2.1945 & 0.0007 & 0.9926 & 2.3891 & 0.0005  \\
\hline
\end{tabular}

    \caption{Prices for a min-butterflies option with $d=2$, $\sigma=0.4$, $r=0.06$, $T=0.5$, $N=10$, $S^1_0=95$ $S^2_0=90$, $K_1=90$, $K_2=110$.  LS implemented with polynomials function of order $6$, $Q_2=Q_3=50000$. The LS price without variance reduction is $2.19$, with a {\bf standard deviation} $\mathbf{0.013}$. \label{tab:minbutter}}
\end{table}

\section{A proxy for the optimal stopping policy}\label{sec:proxy}

All algorithms coupling~\eqref{eq:dpp} and~\eqref{eq:dual-price-k} tend to rely on the optimal policy obtained from the dynamic programming principle to build a dual martingale. Here, we develop a methodology to do it the other way around: use our dual algorithm and the associated martingale $\hat{M}$ to build an optimal policy.
From~\eqref{eq:doob-meyer}, we deduce that for $n \in \{0, \dots, N-1\}$
\begin{equation}
  \label{eq:dual-continuation-value}
  \E[U_{n+1} | \cf_n] =  U_{n+1} + (M^\s_n - M^\s_{n+1}).
\end{equation}
Noting that $\E[Z_{\tau_{n+1}} | \cf_{n}] = \E[\E[Z_{\tau_{n+1}} | \cf_{n+1}] | \cf_{n}] = \E[U_{n+1} | \cf_n]$, the dynamic programming equation for the optimal policy writes
\begin{equation}
  \label{eq:dpp-policy-dual}
  \begin{cases}
    \tau_{N} & = N \\
    \tau_{n} & = n \ind{Z_n \ge U_{n+1} + (M^\s_n - M^\s_{n+1})} + \tau_{n+1} \ind{Z_n < U_{n+1} + (M^\s_n - M^\s_{n+1})}, \quad 0 \le n  \le N-1.
  \end{cases}
\end{equation}
We also have that $U$ solves an other dynamic programming equation
\begin{equation}
  \label{eq:dpp-dual-U}
  \begin{cases}
    U_{N} & = Z_N \\
    U_{n} & = \max(Z_n, U_{n+1} + (M^\s_n - M^\s_{n+1})), \quad 0 \le n  \le N-1.
  \end{cases}
\end{equation}
Note that once $M^\s$ is known, \eqref{eq:dpp-dual-U} is fully explicit and is only useful to solve \eqref{eq:dpp-policy-dual}.

Assume we have $Q$ samples of $(Z^{(q)}_n, M^{\s, (q)}_n)_{0 \le n \le N}$, then we can compute the optimal policy $\tau^{(q)}_1$ along each path $q \in \{1,\dots, Q\}$. So, the primal price can naturally be approximated by Monte Carlo $\hat U^Q_0 = \inv{Q} \sum_{q=1}^Q Z^{(q)}_{\tau^{(q)}_1}$.

From a practical point of view, $M^\s$ is unknown and we can only access some approximation. For instance, we can use the martingale $\hat M$ obtained from~\eqref{eq:hat_M}. However, note that not only \eqref{eq:dual-continuation-value} does not hold anymore if we replace $M^\s$ by $\hat M$, but the quantity $U_{n+1} + (\hat M_n - \hat M_{n+1})$ may not be  even $\cf_n-$~measurable. Hence, the policy obtained using $\hat M$ instead of $M^\s$ may not be a stopping  time anymore. More precisely, let us define the sequence of random times $(\hat{\tau}_n)_{0 \le n \le N}$ by 
\begin{equation}
  \label{eq:dpp-policy-dual-approx}
  \begin{cases}
    \hat{\tau}_{N} & = N \\
    \hat{\tau}_{n} & = n \ind{Z_n \ge \hat{U}_{n+1} + (\hat{M}_n - \hat{M}_{n+1})} + \hat{\tau}_{n+1} \ind{Z_n < \hat{U}_{n+1} + (\hat{M}_n - \hat{M}_{n+1})}, \quad 0 \le n  \le N-1,
  \end{cases}
\end{equation}
where 
\begin{equation*}
  \begin{cases}
    \hat{U}_{N} & = Z_N \\
    \hat{U}_{n} & = \max(Z_n, \hat{U}_{n+1} + (\hat{M}_n - \hat{M}_{n+1})), \quad 0 \le n  \le N-1.
  \end{cases}
\end{equation*}
\begin{proposition}\label{prop_random_time}
  Let $\hat{M}$ be a $\cf$-martingale and $\tau=\inf\{n \ge 0: \hat{U}_n=Z_n \}$. Then, we have $\tau= \hat{\tau}_0$,  $Z_{\tau}-\hat{M}_{\tau}=\max_{0\le n\le N} Z_n-\hat{M}_n$ and
\begin{equation*} 
\E[Z_{\tau}-\hat{M}_{\tau}]=\E[\max_{0\le n\le N} Z_n-\hat{M}_n]\ge U_0.
\end{equation*}
\end{proposition}
\begin{proof}
At first, note that  $\hat{U}$ may not be adapted to the filtration $\cf$. Thus, $(\hat{\tau}_n)_n$ may not define stopping times and $\hat{U}$ may not be a Snell envelope. 
However, we still have $\hat{U}_n-\hat{M}_n=\max(\hat{Z}_n-\hat{M}_n, \hat{U}_{n+1} - \hat{M}_{n+1})$ and thus by induction \begin{equation}\label{eq:max_courant}
\hat{U}_n-\hat{M}_n=\max_{n\le p\le N} Z_p-\hat{M}_p.
\end{equation}
We then check that $\hat{\tau}_n = \inf \{p\ge n:  \hat{U}_p-\hat{M}_p=Z_p-\hat{M}_p \}$ by backward induction. This is true for $n=N$. For $n<N$, if $\hat{\tau}_n>n$, then from~\eqref{eq:dpp-policy-dual-approx} we have $Z_n- \hat{M}_n<\hat{U}_{n+1}-\hat{M}_{n+1}$ and $\hat{\tau}_n=\hat{\tau}_{n+1}$ which gives the claim by the induction assumption. Otherwise, $\hat{\tau}_n=n$ and thus  $Z_n- \hat{M}_n\ge  \hat{U}_{n+1}-\hat{M}_{n+1}$, which gives  $Z_n- \hat{M}_n =  \hat{U}_{n}-\hat{M}_{n}$ by~\eqref{eq:max_courant}.  Therefore, 
$\tau=\inf\{n\ge 0 : \hat{U}_n-\hat{M}_n=Z_n-\hat{M}_n \}=\hat{\tau}_0$.
Equation~\eqref{eq:max_courant} then gives $Z_{\tau}-\hat{M}_{\tau}\ge Z_p-\hat{M}_p$ for $p\ge \tau$. By the last formula for $\tau$, we have $\hat{U}_{n}-\hat{M}_n=\hat{U}_{n+1}-\hat{M}_{n+1}$ for $n<\tau$ and thus $\hat{U}_\tau-\hat{M}_\tau=\hat{U}_{n}-\hat{M}_n>Z_n-\hat{M}_n$. Therefore we have
$$Z_{\tau}-\hat{M}_{\tau}=\max_{0\le n\le N} Z_n-\hat{M}_n,$$
and  the claim follows by the dual formulation~\eqref{eq:dual-price-k}. 
\end{proof}Note that Proposition~\ref{prop_random_time} holds for any $\cf$-martingale $\hat{M}$. When we take for $\hat{M}$ the approximation of $M^\s$ given by the dual algorithm~\eqref{eq:hat_M}, then $\E[Z_{\hat{\tau}_0}-\hat{M}_{\hat{\tau}_0}]$ coincides with the obtained dual price given by this algorithm. 
In the particular case when $\hat{M}=M^\s$, we have $\hat{\tau}_0=\tau^\s$ and therefore $\hat{\tau}_0-\tau^\s$ can be seen in general as a measure of the quality of the approximating martingale~$\hat{M}$. Of~course, $\tau^\s$ is not known exactly, but we use a Longstaff Schwartz algorithm to approximate it. Thus, in Figure~\ref{fig:diff-tau}, we plot the histograms of $\hat{\tau}_0-\hat{\tau}^{LS_1,Q_2}$ for two different martingales $\hat{M}$ on the one dimensional Put option example of Subsection~\ref{sec:var-1d}. We observe, as we may expect, that $\P(\hat{\tau}_0=\hat{\tau}^{LS_1,Q_2})$ is relatively high and more generally the random time $\hat{\tau}_0$ is close to the stopping time $\hat{\tau}^{LS_1,Q_2}$.

\begin{figure}
  \includegraphics[width=0.45\linewidth]{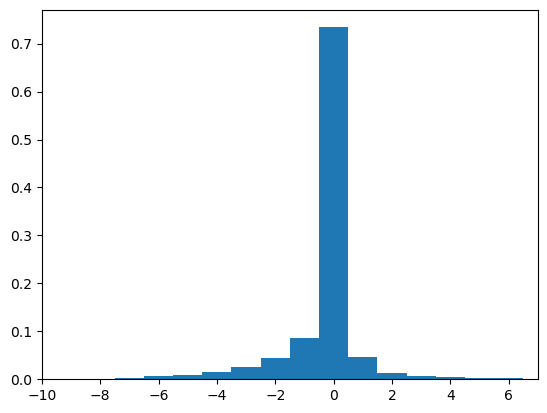} \includegraphics[width=0.45\linewidth]{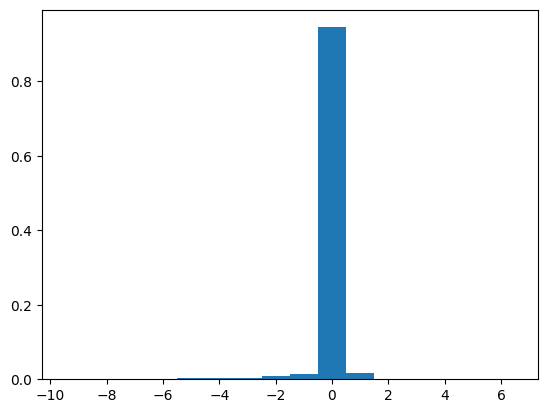}
  \caption{Histogram of $\hat{\tau}_0-\hat{\tau}^{LS_1,Q_2}$ on the Put option example obtained with $50000$ samples, $P=50$ and $Q_2=50000$, using the martingale $\hat{M}$ obtained with $\bar{N}=20$ subticks (left, with $Q_1=2\times 10^6$) or using the martingale $\hat{M}$ that includes the European call option (right, with $Q_1=10^5$), see Equation~\eqref{def_mg_incr}. }\label{fig:diff-tau}
\end{figure}

\appendix

\section{Variance decomposition in the Longstaff Schwartz algorithm}
\label{app:variance}

From \citet[Theorem~4.5]{LEL19} or~\citet[Theorem~4.10]{lapeyre:hal-02183587}, we get the following results on the convergence of the estimator~\eqref{eq:prix_empirique} when $Q_2=Q_3$. First, for every $n=1,\dots,N$,
\begin{equation}
  \label{eq:empirical-var}
  \lim_{Q \to \infty} \quad \inv{Q} \sum_{q=1}^Q \left(Z_{\hat{\tau}_n^{LS_2, Q, (q)}}^{(q)}\right)^2 - \left(\inv{Q} \sum_{q=1}^Q Z_{\hat{\tau}_n^{LS_2, Q, (q)}}^{(q)} \right)^2 = \Var(Z_{\tau_n^{LS_2}}) \quad a.s.
\end{equation}
The convergence rate analysis carried out in~\cite{CLP} applies steadily to our approach. Then, under suitable assumptions, the vector
\begin{equation*}
  \left(\sqrt{Q} \left(\inv{Q} \sum_{q=1}^Q Z_{\hat{\tau}_n^{LS_2, Q, (q)}}^{(q)} - \E[Z_{\tau^{LS_2}_n}] \right)\right)_{n=1,\dots, N}
\end{equation*}
converges in law to a Gaussian vector. As noted in~\citet[Theorems~4.2 and~4.3]{CLP}, determining the asymptotic variance directly from the data generated by a single run of the algorithm is almost impossible. From the proof of the central limit theorem for their algorithm, we have that
\begin{equation*}
  \sqrt{Q} \left(\inv{Q} \sum_{q=1}^Q Z_{\hat{\tau}_n^{LS_2, Q, (q)}}^{(q)} - \E[Z_{\tau_n^{LS_2}}] \right) \xrightarrow[Q \to \infty]{\mathcal{L}} G_1+G_2,
\end{equation*}
where $(G_1,G_2)$ is a Gaussian vector. We have $G_1\sim \mathcal{N}(0,\Var(Z_{\tau_n^{LS_2}}))$ as the limit of $\sqrt{Q} \left(\inv{Q} \sum_{q=1}^Q   Z_{\tau_n^{LS_2, (q)}}^{(q)} -\E[Z_{\tau_n^{LS_2}}]  \right)$  by the standard central limit theorem. The variance of $G_2$  is related to the approximation of $\tau_n^{LS_2}$ by $\hat{\tau}_n^{LS_2,Q}$. Then, using the empirical variance of the estimator as a measurement of the algorithm convergence misses part of the variance since from~\eqref{eq:empirical-var}, we know that the empirical variance only takes into account $\Var(G_1)$. Yet, on our numerical examples, the predominant part of the variance comes from $G_1$: in the put example of Table~\ref{Table:Put_local}, we have obtained $\Var(G_1)=\Var(Z_{\tau_n^{LS_2}})\approx 0.0506$ while our estimation of the whole variance is $0.0513$.

\section{Calculation of $\hat{M}$: the algorithm}\label{Sec:CalcM}

First, we rewrite~\eqref{eq:dual-price-k} as follows
\begin{align}
  U_0&=  \E[Z_N] + \inf_{\Delta M_1 \in \ch_1^2,\dots,\Delta M_N \in \ch_N^2 }  \sum_{n=0}^{N-1}\E\left[ \left(Z_{n} + \Delta M_{n+1} - \max_{n + 1\le j \le N} \left\{Z_{j} - \sum_{i=n+2}^{j} \Delta M_i\right\} \right)_+\right], \label{global_min}
\end{align}
where $\Delta M_n=M_n-M_{n-1}$ and
\[
  \ch^2_n = \{ Y \in \L^2(\Omega) \,:\, Y \text{ is real valued, }\cf_{n}-\text{measurable and } \E[Y | \cf_{n-1}] = 0\}.
\]
The aim of~\eqref{global_min} is to decompose the global optimization problem into a sequence of smaller ones. However, the lack of strict convexity of $(\cdot)_+$ does not allow to make this rigorous as well as to have a unique minimum. In~\cite{AKL}, we consider the same minimisation problem, when the positive part is replaced by a strictly convex function and obtain in particular the following result. 

\begin{theorem}\label{thm:main}
   Let  $M^\s \in \H^2$ be defined by~\eqref{eq:doob-meyer}.  Then, we have for any $M\in \H^2$
 \begin{align*}
  &\E[Z_N] + \sum_{n=0}^{N-1} \E\left[ \left(Z_{n} + \Delta M_{n+1} - \max_{n + 1\le j \le N} \left\{Z_{j} - \sum_{i=n+2}^{j} \Delta M_i\right\} \right)^2\right]\\
  &\ge \E[Z_N] + \sum_{n=0}^{N-1} \E\left[  \left(Z_{n} + \Delta M^\s_{n+1} - \max_{n + 1\le j \le N} \left\{Z_{j} - \sum_{i=n+2}^{j} \Delta M^\s_i\right\} \right)^2\right],\notag
 \end{align*}  
 and $M^\s$ is the unique solution of the following sequence of backward optimization problems for $n=N-1,\dots,0$
 \begin{equation} \label{eq:inf}
 \inf_{\Delta M_{n+1} \in \ch^2_{n+1}}\E\left[  \left(Z_{n} + \Delta M_{n+1}  - \max_{n + 1\le j \le N} \left\{Z_{j} - \sum_{i=n+2}^{j} \Delta M_i\right\}\right)^2\right].
 \end{equation}
 \end{theorem}

 From a practical point of view, using Theorem~\ref{thm:main} to approximate $M^\s$ requires to consider finite dimensional approximations $\ch^{2, pr}_n$ of $\ch^2_n$ for each~$n\in \{1,\dots,N\}$ and to replace the expectations by sample averages to effectively solve the least squares problem.

 We assume that the sub-vector spaces $\ch^{2, pr}_n$, $1\le n \le N$, are spanned by $L \in \N^*$ random variables $\Delta X_{n,\ell} \in \ch^2_n$, $1\le \ell \le L$:
$$\ch^{pr}_n=\Span\left(\left\{ \alpha \cdot \Delta X_{n} \ : \ \alpha \in \R^L \right\}\right).$$
Here, $L$ does not depend on $n$ for simplicity, and we write $\alpha \cdot \Delta X_{n}= \sum_{\ell=1}^L \alpha_\ell \Delta X_{n,\ell}$. Then, the minimisation problem~\eqref{eq:inf} becomes, for $0\le n \le N-1$,
\begin{equation*}
  \inf_{\alpha \in \R^L }\E\left[ \left(Z_{n} + \alpha \cdot \Delta X_{n+1}  - \max_{n + 1\le j \le N} \left\{Z_{j} - \sum_{i=n+2}^{j} \Delta M_i\right\}\right)^2\right].
\end{equation*} 
This is a standard least squares optimisation problem: if the positive semidefinite matrix $\E[\Delta X_{n+1}\Delta X_{n+1}^T]$ is invertible, the minimum is given by
\begin{align}
   \alpha_{n+1} &= \left( \E[\Delta X_{n+1}\Delta X_{n+1}^T] \right)^{-1} \E\left[ \left( \max_{n + 1\le j \le N} \left\{Z_{j} - \sum_{i=n+2}^{j} \Delta M_i\right\} -Z_n \right) \Delta X_{n+1}\right] \notag\\
   &= \left( \E[\Delta X_{n+1}\Delta X_{n+1}^T] \right)^{-1} \E\left[ \left( \max_{n + 1\le j \le N} \left\{Z_{j} - \sum_{i=n+2}^{j} \Delta M_i\right\} \right) \Delta X_{n+1}\right], \label{def_alpha_n}
\end{align}
since $\E[\Delta X_{n+1}|\cf_n]=0$. 

Assume that we have $Q$ independent paths $Z^{(q)}_n$ of the underlying process $Z_n$ and $\Delta X^{(q)}_n$ of the martingale increments $\Delta X_n$, for $1\le n \le N$ and $1\le q\le Q$. 
If the matrix $\sum_{q=1}^Q \Delta X^{(q)}_{n+1} (\Delta X^{(q)}_{n+1})^T$ is invertible, then we define $\alpha^Q_{n+1}$ as the Monte Carlo estimator of~\eqref{def_alpha_n} by
\begin{equation}\label{def_alpha_n_Q}
  \alpha_{n+1}^Q=\left( \sum_{q=1}^Q \Delta X^q_{n+1} (\Delta X^q_{n+1})^T \right)^{-1} \sum_{q=1}^Q \max_{n + 1\le j \le N} \left\{ Z^q_{j} - \sum_{i=n+2}^{j} \alpha^Q_i \cdot \Delta X^q_i \right\} \Delta X^q_{n+1}
\end{equation}

From~\cite[Proposition~3]{AKL}, $\alpha_{n}^Q\to_{Q\to \infty} \alpha_n$ a.s., so it is natural to consider the following approximation of $M^\s$ 
\begin{equation}
  \label{eq:hat_M}
  \hat M_n = \sum_{\ell=1}^{n} \alpha^Q_\ell \cdot \Delta X_n.
\end{equation}
Since perfect hedging is only attainable in continuous time when a martingale representation theorem holds, we suggest to use a sub-grid to parametrise the martingale increments. Each interval $[T_i, T_{i+1}]$ for $0 \le i \le  N-1 $ is split into $\bar{N}$ regular sub-intervals, and we set 
\begin{equation*}
   t_{i,j}=T_i+\frac{j}{\bar{N}} \frac{T}{N}, \text{ for }0 \le j \le \bar{N}.\end{equation*}
We consider a family of functions $u_p:\R^d \to \R$ for $p\in \{1,\dots,\bar{P}\}$ and a family of discounted assets $(\ca^k)_{1\le k\le \bar{d}}$. Then, we define the  following elementary martingale increments:
\begin{equation}\label{def_mg_incr}
  X^{{p},k}_{t_{i,j}} - X^{{p},k}_{t_{i,j-1}} =  u^{p}_{i,j-1} (S_{t_{i,j-1}}) (\ca^k_{t_{i,j}}- \ca^k_{t_{i,j-1}}).
\end{equation}
 In this paper, we consider for $u^p_{i,j}$ indicator functions with $P$ distinct intervals for each dimension (local basis) and martingales constructed only with the underlying asset i.e.  $\bar{d}=d$ and  $\ca^k=\tilde{S}^k$. Other choices with polynomial functions for regression and  European option claims in the hedging martingale are presented in~\cite{AKL}.

\bibliographystyle{abbrvnat}
\bibliography{biblio.bib}
\end{document}